\documentclass[12pt]{article}

\usepackage{amsmath,amsthm,amssymb,amscd,epsfig,url}
\usepackage{fullpage}
\usepackage{times}
\newtheorem{theorem}{Theorem}[section]
\newtheorem{lemma}[theorem]{Lemma}
\newtheorem{proposition}[theorem]{Proposition}
\newtheorem{corollary}[theorem]{Corollary}

\newtheorem{definition}[theorem]{Definition}

\def\<{{\langle}} \def\>{{\rangle}}       
\title{A Note on Amortized Branching Program Complexity}
\author{
Aaron Potechin \\
Institute for Advanced Study
}
\date{\today}

\begin{document}
\maketitle

\begin{abstract}
In this paper, we show that while almost all functions require exponential size branching programs to compute, for all functions $f$ there is a branching program computing a doubly exponential number of copies of $f$ which has linear size per copy of $f$. This result disproves a conjecture about non-uniform catalytic computation, rules out a certain type of bottleneck argument for proving non-monotone space lower bounds, and can be thought of as a constructive analogue of Razborov's result that submodular complexity measures have maximum value $O(n)$.
\end{abstract}

\thispagestyle{empty}
\noindent \textbf{Acknowledgement: This material is based upon work supported by the National Science Foundation under agreement No. CCF-1412958 and by the Simons Foundation. The author would like to thank Avi Wigderson for helpful conversations.}\\
.\newpage

\section{Introduction}
In amortized analysis, which appears throughout complexity theory and algorithm design, rather than considering the worst case cost of an operation, we consider the average cost of the operation when it is repeated many times. This is very useful in the situation where operations may have a high cost but if so, this reduces the cost of future operations. In this case, the worst-case rarely occurs and the average cost of the operation is much lower. A natural question we can ask is as follows. Does amortization only help for specific operations, or can any operation/function be amortized? 

For boolean circuits, which are closely related to time complexity, Uhlig \cite{uhligone},\cite{uhligtwo} showed that for any function $f$, as long as $m$ is $2^{o(\frac{n}{\log{n}})}$ there is a circuit of size $O(\frac{2^n}{n})$ computing $f$ on $m$ different inputs simultaneously. As shown by Shannon \cite{shannon} and Lyupanov \cite{lyupanov}, almost all functions require circuits of size $\Theta(\frac{2^n}{n})$ to compute, which means that for almost all functions $f$, the cost to compute many inputs of $f$ is essentially the same as the cost to compute one input of $f$!

In this paper, we consider a similar question for branching programs, which are closely related to space complexity. In particular, what is the minimum size of a branching program which computes many copies of a function $f$ on the same input? This question is highly non-trivial because branching programs are not allowed to copy bits, so we cannot just compute $f$ once and then copy it. In this paper, we show that for $m = 2^{2^n-1}$, there is a branching program computing $m$ copies of $f$ which has size $O(mn)$ and thus has size $O(n)$ per copy of $f$.

This work has connections to several other results in complexity theory. In catalytic computation, introduced by Buhrman, Cleve, Kouck\'{y}, Loff, and Speelman \cite{catalyticspace}, we have an additional tape of memory which is initially full of unknown contents. We are allowed to use this tape, but we must restore it to its original state at the end of our computation. As observed by Girard, Kouck\'{y}, and McKenzie \cite{nonuniformcatalytic}, the model of a branching program computing multiple instances of a function is a non-uniform analogue of catalytic computation and our result disproves Conjecture 25 of their paper. Our result also rules out certain approaches for proving general space lower bounds. In particular, any lower bound technique which would prove a lower bound on amortized branching program complexity as well as branching program size cannot prove non-trivial lower bounds. Finally, our result is closely related to Razborov's result \cite{submodular} that submodular complexity measures have maximum size $O(n)$ and can be thought of as a constructive analogue of Razborov's argument.

\subsection{Outline}
In Section \ref{prelim} we give some preliminary definitions. In section \ref{construct} we give our branching program construction, proving our main result. In section \ref{catalytic} we briefly describe the relationship between our work and catalytic computation. In section \ref{lowerboundbarrier} we discuss which lower bound techniques for proving general space lower bounds are ruled out by our construction. Finally, in section \ref{relationtosubmodular} we describe how our work relates to Razborov's result \cite{submodular} on submodular complexity measures. 
\section{Preliminaries}\label{prelim}
In this section, we define branching programs, branching programs computing multiple copies of a function, and the amortized branching program complexity of a function.
\begin{definition}
We define a branching program to be a directed acyclic multi-graph $G$ with labeled edges and distinguished start nodes, accept nodes, and reject nodes which satisfies the following conditions.
\begin{enumerate}
\item Every vertex of $G$ has outdegree 0 or 2. For each vertex $v \in V(G)$ with outdegree 2, there exists an $i \in [1,n]$ such that one of the edges going out from $v$ has label $x_i = 0$ and the other edge going out from $v$ has label $x_i = 1$.
\item Every vertex with outdegree 0 is an accept node or a reject node.
\end{enumerate}
Given a start node $s$ of a branching program and an input $x \in \{0,1\}^n$, we start at $s$ and do the following at each vertex $v$ that we reach. If $v$ is an accept or reject node then we accept or reject, respectively. Otherwise, for some $i$, one of the labels going out from $v$ has label $x_i = 0$ and the other edge going out from $v$ has label $x_i = 1$. If $x_i = 0$ then we take the edge with label $x_i = 0$ and if $x_i = 1$ then we take the edge with label $x_i = 1$. In other words, we follow the path starting at $s$ whose edge labels match $x$ until we reach an accept or reject node and accept or reject accordingly.

Given a branching program $G$ and start node $s$, we define the function $f_s$ so that $f_s(x) = 1$ if we reach an accept node when we start at $s$ on input $x$ and $f_s(x) = 0$ if we reach a reject node when we start at $s$. We say that $(G,s)$ computes the function $f_s$.

We define the size of a branching program $G$ to be $|V(G)|$, the number of vertices/nodes of $G$.
\end{definition}
\begin{definition}
We say that a branching program computes $f$ $m$ times if the following is true:
\begin{enumerate}
\item The branching program has $m$ start nodes $s_1,\cdots,s_m$, $m$ accept nodes $a_1,\cdots,a_m$, and $m$ reject nodes $r_1,\cdots,r_m$
\item For all $i$, if the branching program starts at $s_i$ then on input $x$ it will end at $a_i$ if $f(x)=1$ and it will end at $r_i$ if $f(x) = 0$
\end{enumerate}
\end{definition}
\begin{definition}
Given a function $f$,
\begin{enumerate}
\item We define $b_m(f)$ to be the minimal size of a branching program which computes $f$ $m$ times.
\item We define the amortized branching program complexity $b_{avg}(f)$ of $f$ to be \\
$b_{avg}(f) = \lim_{m \to \infty}{\frac{b_m(f)}{m}}$
\end{enumerate}
\end{definition}
\begin{proposition}
For all functions $f$, $b_{avg}(f)$ is well-defined and is equal to $\inf{\{\frac{b_m(f)}{m}:m \geq 1\}}$ 
\end{proposition}
\begin{proof}
Note that for all $m_1,m_2 \geq 1$, $b_{m_1 + m_2}(f) \leq b_{m_1}(f) + b_{m_2}(f)$ as if we are given a branching program computing $f$ $m_1$ times and a branching program computing $f$ $m_2$ times, we can take their disjoint union and this will be a branching program computing $f$ $m_1 + m_2$ times. Thus for all $m_0 \geq 1$, for all $k \geq 1,0 \leq r < m_0$, $b_{k{m_0}+r}(f) \leq kb_{m_0}(f)+b_r(f)$. This implies that $\lim_{m \to \infty}{\frac{b_{m}(f)}{m}} \leq \frac{b_{m_0}(f)}{m_0}$ and the result follows.
\end{proof}
\section{The Construction}\label{construct}
In this section, we give our construction of a branching program computing doubly exponentially many copies of a function $f$ which has linear size per copy of $f$, proving our main result.
\begin{theorem}\label{construction}
For all $f$, $b_{avg}(f) \leq 64n$. In particular, for all $f$, taking $m = 2^{2^n-1}$, $b_{m}(f) \leq 32n2^{2^n}$
\end{theorem}
\begin{proof}
Our branching program has several parts. We first describe each of these parts and how we put them together and then we will describe how to construct each part. The first two parts are as follows:
\begin{enumerate}
\item A branching program which simultaneously identifies all functions $g:\{0,1\}^n \to \{0,1\}$ that have value $1$ for a given $x$. More preceisely, it has start nodes $s_1,\cdots,s_m$ where $m = 2^{2^n-1}$ and has one end node $t_g$ for each possible function $g:\{0,1\}^n \to \{0,1\}$, with the guarantee that if $g(x) = 1$ for a given $g$ and $x$ then there exists an $i$ such that that the branching program goes from $s_i$ to $t_g$ on input $x$.
\item A branching program which simultaneously evaluates all functions $g:\{0,1\}^n \to \{0,1\}$. More precisely, it has one start node $s_g$ for each function $g$ and has end nodes $a'_1,\cdots,a'_m$ and $r'_1,\cdots,r'_m$, with the guarantee that for a given $g$ and $x$, if $g(x) = 1$ then the branching program goes from $s_g$ to $a'_i$ for some $i$ and if $g(x) = 0$ then the branching program goes from $s_g$ to $r'_i$ for some $i$. 
\end{enumerate}
If $f$ is the function which we actually want to compute, we combine these two parts as follows. The first part gives us paths from $\{s_i:i \in [1,m]\}$ to $\{t_g:g(x) = 1\}$. We now take each $t_g$ from the first part and set it equal to $s_{(f \wedge g) \vee (\neg{f} \wedge \neg{g})}$ in the second part. Once we do this, if $f(x) = 1$ then for all $g$, $g(x) = 1 \iff (f \wedge g) \vee (\neg{f} \wedge \neg{g}) = 1$ so we will have paths from $\{t_g:g(x) = 1\} = \{s_{(f \wedge g) \vee (\neg{f} \wedge \neg{g})}:g(x)=1\}$ to $\{a_i:i \in [1,m]\}$. If $f(x) = 0$ then for all $g$, $g(x) = 1 \iff (f \wedge g) \vee (\neg{f} \wedge \neg{g}) = 0$, so we will have paths from $\{t_g:g(x) = 1\} = \{s_{(f \wedge g) \vee (\neg{f} \wedge \neg{g})}:g(x)=1\}$ to $\{r_i:i \in [1,m]\}$. Putting everything together, when $f(x) = 1$ we will have paths from $\{s_i:i \in [1,m]\}$ to $\{a_i:i \in [1,m]\}$ and when $f(x) = 0$ we will have paths from $\{s_i:i \in [1,m]\}$ to $\{r_i:i \in [1,m]\}$.

However, these paths do not have to map $s_i$ to $a'_i$ or $r'_i$, they can permute the final destinations. To fix this, our final part will run the branching program we have so far in reverse. This fixes the permutation issue but gets us right back where we started! To avoid this, we have two copies of this final part, one applied to $\{a'_i:i \in [1,m]\}$ and one applied to $\{r'_i:i \in [1,m]\}$. This separates the case when $f(x) = 1$ and the case $f(x) = 0$, giving us our final branching program.

We now describe how to construct each part. For the first part, which simultaneously identifies the functions which have value $1$ on input $x$, we have a layered branching program with $n+1$ levels going from $0$ to $n$. At level $j$, for each function $g:\{0,1\}^{j} \to \{0,1\}$, we have $2^{2^n-2^j}$ nodes corresponding to $g$. For all $j \in [1,n]$ we draw the arrows from level $j-1$ to level $j$ as follows. For a node corresponding to a function $g:\{0,1\}^{j-1} \to \{0,1\}$, we draw an arrow with label $x_j=1$ from it to a node corresponding to a function $g':\{0,1\}^{j} \to \{0,1\}$ such that $g'(x_1,\cdots,x_{j-1},1)=g(x_1,\cdots,x_{j-1})$. Similarly, we draw an arrow with label $x_j=0$ from it to a node corresponding to a function $g':\{0,1\}^{j} \to \{0,1\}$ such that $g'(x_1,\cdots,x_{j-1},0)=g(x_1,\cdots,x_{j-1})$. We make these choices arbitrarily but make sure that no two arrows with the same label have the same destination.
\begin{figure}[ht]
\centerline{\includegraphics[height=8cm]{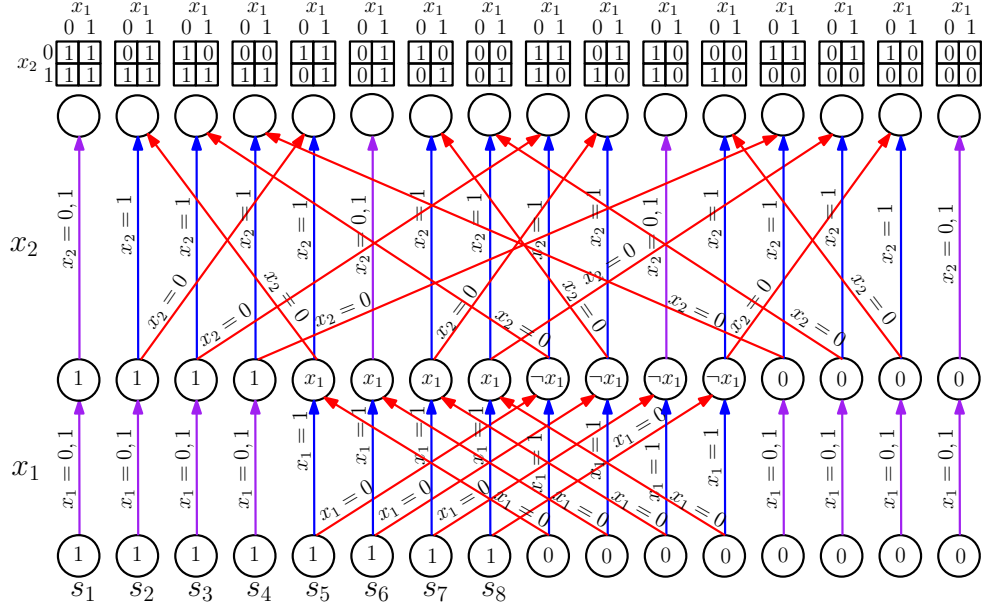}}
\caption{This figure illustrates part 1 of our construction for $n=2$. The functions for the top vertices are given by the truth tables at the top and each other vertex corresponds to the function inside it. Blue edges can be taken when the corresponding variable has value $1$, red edges can be taken when the corresponding variable has value $0$, and purple edges represent both a red edge and a blue edge (which are drawn as one edge to make the diagram cleaner). Note that for all inputs $x$, there are paths from the start nodes to the functions which have value $1$ on input $x$ at each level.}
\label{identifying}
\end{figure}

For the second part, which simultaneously evaluates each function, we have a layered branching program  with $n+1$ levels going from $0$ to $n$. At level $n-j$, for each function $g:\{0,1\}^{j} \to \{0,1\}$, we have $2^{2^n-2^{j}}$ nodes corresponding to $g$. For all $j \in [1,n]$ we draw the arrows from level $n-j$ to level $n-j+1$ as follows. For a node corresponding to a function $g:\{0,1\}^{j} \to \{0,1\}$, we draw an arrow with label $x_j=1$ from it to a node corresponding to the function $g(x_1,\cdots,x_{j-1},1)$ and draw an arrow with label $x_j=0$ from it to a node corresponding to the function $g(x_1,\cdots,x_{j-1},0)$. Again, we make these choices arbitrarily but make sure that no two edges with the same label have the same destination.
\begin{figure}[ht]
\centerline{\includegraphics[height=8cm]{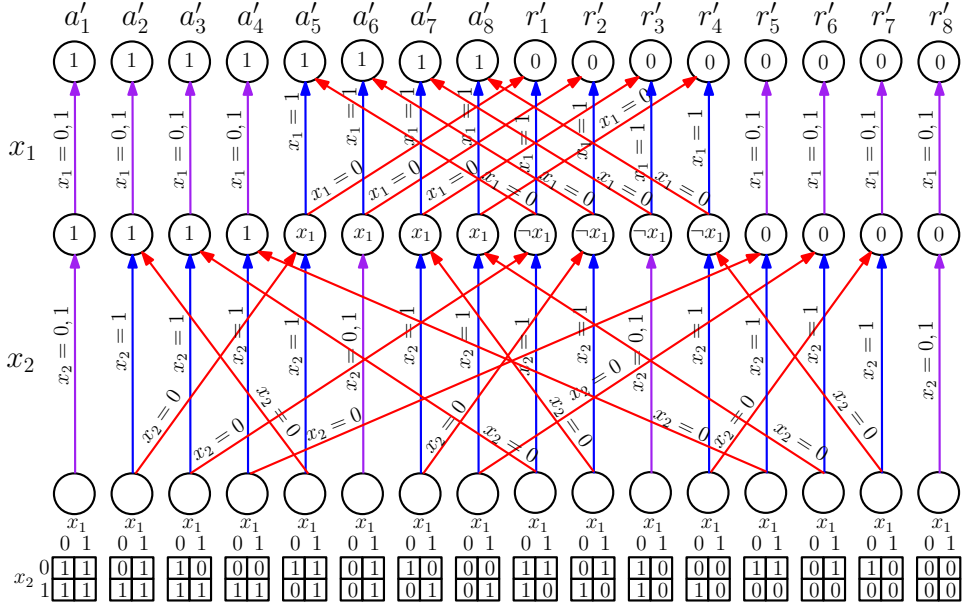}}
\caption{This figure illustrates part 2 of our construction for $n=2$. The functions for the bottom vertices are given by the truth tables at the bottom and each other vertex corresponds to the function inside it. Note that for all inputs $x$, the paths go between the functions which evaluate to $1$ on input $x$ and the accept nodes and between the functions which evaluate to $0$ on input $x$ and the reject nodes.}
\label{identifying}
\end{figure}

For the final part, note that because we made sure not to have any two edges with the same label have the same destination and each level has the same number of nodes, our construction so far must have the following properties
\begin{enumerate}
\item Every vertex has indegree 0 or 2. For the vertices $v$ with indegree 2, there is a $j$ such that one edge going into $v$ has label $x_j = 1$ and the other edge going into $v$ has label $x_j = 0$.
\item The vertices which have indegree 0 are precisely the vertices in the bottom level.
\end{enumerate}
These conditions imply that if we reverse the direction of each edge in the branching program we have so far, this gives us a branching program which runs our branching program in reverse. As described before, we now take two copies of this reverse program. For one copy, we take its start nodes to be $a'_1,\cdots,a'_m$ and relabel its copies of $s_1,\cdots,s_m$ as $a_1,\cdots,a_m$. For the other copy, we take its start nodes to be $r'_1,\cdots,r'_m$ and relabel its copies of $s_1,\cdots,s_m$ as $r_1,\cdots,r_m$.
\end{proof}
\section{Relationship to catalytic computation}\label{catalytic}
In catalytic computation, we have additional memory which we may use but this memory starts with unknown contents and we must restore this memory to its original state at the end. Our result is related to catalytic computation through Proposition 9 of \cite{nonuniformcatalytic}, which says the following
\begin{proposition}
Let $f$ be a function which can be computed in space $s(n)$ using catalytic tape of size $l(n) \leq 2^{s(n)}$. Then $b_{2^l}(f)$ is ${2^l} \cdot 2^{O(s(n))}$.
\end{proposition}
For convenience, we give a proof sketch of this result here.
\begin{proof}[Proof sketch]
This can be proved using the same reduction that is used to reduce a Turing machine using space $s(n)$ to a branching program of size $2^{O(s(n))}$, with the following differences. There are $2^l$ possibilities for what is in the catalytic tape at any given time, so the resulting branching program is a factor of $2^l$ times larger. The requirement that the catalytic tape is restored to its original state at the end implies that there must be $2^l$ disjoint copies of the start, accept, and reject nodes, one for each possibility for what is in the catalytic tape originally. This means that the branching program computes $f$ $2^l$ times. Finally, the condition that $l(n) \leq 2^{s(n)}$ is necessary because otherwise the branching program would have to be larger in order to keep track of where the pointer to the catalytic tape is pointing!
\end{proof}
Girard, Kouck\'{y}, and McKenzie \cite{nonuniformcatalytic} conjectured that for a random function $f$, for all $m \geq 1$, $b_m(f)$ is $\Omega(mb_1(f))$. If true, this conjecture would imply (aside from issues of non-uniformity) that a catalytic tape does not significantly reduce the space required for computing most functions. However, our construction disproves this conjecture. 

That said, our construction requires $m$ to be doubly exponential in $n$. It is quite possible that $\log(\frac{b_m(f)}{m})$ is $\Omega(\log(b_1(f)))$ for much smaller $m$, which would still imply (aside from issues of non-uniformity) that a catalytic tape does not reduce the space required for computing most functions by more than a constant factor.
\section{Barrier for input-based bottleneck arguments}\label{lowerboundbarrier}
As noted in the introduction, our result rules out any general lower bound approach which would prove lower bounds on amortized branching program complexity as well as branching program size. In this section, we discuss one such class of techniques. 

One way we could try to show lower bounds on branching programs is as follows. We could argue that for the given function $f$ and a given branching program $G$ computing $f$, for every YES input $x$ the path that $G$ takes on input $x$ contains a vertex giving a lot of information about $x$ and thus $G$ must be large to accomodate all of the possible inputs. We observe that this kind of argument would show lower bounds on amortized branching prorgam complexity as well as on branching program size and thus cannot show nontrivial general lower bounds. In particular, we show the following.
\begin{lemma}\label{flowlemma}
Assume that we have a function $f$ and a set of YES inputs $I$ of $f$. If there is a criterion for assigning vertices of a branching program to inputs in $I$ such that
\begin{enumerate}
\item For all $x \in I$ and any path in a branching program from a start node to an accept node on input $x$, there is a vertex $v_x$ in this path which can be assigned to $x$.
\item For any vertex $v$ in any branching program computing $f$, $v$ can be assigned to at most $\frac{1}{S}$ of the inputs in $I$
\end{enumerate}
then $b_m(f) \geq mS$
\end{lemma}
\begin{proof}
Let $G$ be a branching program computing $f$ $m$ times. The total number of times some vertex in $G$ is assigned to an input in $I$ is as least $m|I|$. However, each vertex of the branching program can be assigned to at most $\frac{|I|}{S}$ inputs in $I$, so the total number of times a vertex is assigned to an input in $I$ is at most $\frac{|V(G)|\cdot|I|}{S}$, where $|V(G)|$ is the size of the branching program. Thus, $\frac{|V(G)|\cdot|I|}{S} \geq m|I|$ which implies that $|V(G)| \geq mS$, as needed.
\end{proof}
If we could give such a criterion for some function $f$ with a large $S$, this would imply that $b_1(f) \geq S$, giving us a large lower bound on the size of any branching program computing $f$. However, Theorem \ref{construction} says that for $m = 2^{2^n-1}$, $b_m(f) \leq 64mn$. Thus, $mS \leq 64mn$, so $S \leq 64n$ and so no such argument can prove a superlinear lower bound. 

In fact, the reasoning behind Lemma \ref{flowlemma} proves a lower bound not only on $b_m(f)$, but on the size of any branching program $G$ with $m$ start nodes $s_1,\cdots,s_m$ where for every start node $s_i$, $(G,s_i)$ computes the function $f$ and the computation paths are disjoint for any given input $x$. In other words, it does not matter if there is a permutation of which accept or reject nodes are reached. For the purposes of Lemma \ref{flowlemma}, we only need the first two parts of the construction in Theorem \ref{construction} so we have the same upper bound on $S$ even for oblivious read-twice branching programs!

Thus, lemma \ref{flowlemma} provides an explanation for the spike in difficulty between proving lower bounds for read-once branching programs and read-twice branching programs which can be seen in Razborov's survey \cite{razborovsurvey}. Lemma \ref{flowlemma} also implies that the current framework of Potechin and Chan \cite{switchingnetwork} for analyzing monotone switching networks cannot prove general lower bounds without being modified significantly, as it currently uses an input-based bottleneck argument. However, this construction does not say anything about lower bounds based on counting functions such as Neciporuk's quadratic lower bound \cite{neciporuk} or lower bounds based on communication complexity arguments.

\section{Linear upper bound on complexity measures}\label{relationtosubmodular}
Another way we could try to lower bound branching program size is through a complexity measure on functions. However, Razborov \cite{submodular} showed that submodular complexity measures cannot have superlinear values. In this section we show that this is also true for a similar class of complexity measures, branching complexity measures, which correspond more closely to branching programs. We then show that all submodular complexity measures are also branching complexity measures, so Theorem \ref{construction} is a constructive analogue and a slight generalization of Razborov's result \cite{submodular}.
\begin{definition}
We define a branching complexity measure $\mu_b$ to be a measure on functions which satisfies the following properties
\begin{enumerate}
\item $\forall i, \mu_b(x_i) = \mu_b(\neg{x_i}) = 1$
\item $\forall f, \mu_b(f) \geq 0$
\item $\forall f,i, \mu_b(f \wedge x_i) + \mu_b(f \wedge \neg{x_i}) \leq \mu_b(f) + 2$
\item $\forall f,g, \mu_b(f \vee g) \leq \mu_b(f) + \mu_b(g)$
\end{enumerate}
\end{definition}
\begin{definition}
Given a node $v$ in a branching program, define $f_v(x)$ to be the function such that $f_v(x)$ is $1$ if there is a path from some start node to $v$ on input $x$ and $0$ otherwise. Note that for any start node $s$, $f_s = 1$.
\end{definition}
\begin{lemma}\label{branchingmeasure}
If $\mu_b$ is a branching complexity measure then for any branching program, the number of non-end nodes which it contains is at least $\frac{1}{2}\left(\sum_{t:t \text{ is an end node}}{\mu_b(f_t)} - \sum_{s:s \text{ is a start node}}{\mu_b(f_s)}\right)$.
\end{lemma}
\begin{proof}
To see this, consider what happens to $\sum_{t:t \text{ is an end node}}{\mu_b(f_t)} - \sum_{s:s \text{ is a start node}}{\mu_b(f_s)}$ as we construct the branching program. At the start, when we only have the start nodes and these are also our end nodes, this expression has value $0$. Each time we merge end nodes together, this can only decrease this expression. Each time we branch off from an end node, making the current node a non-end node and creating two new end nodes, this expression increases by at most $2$. Thus, the final value of this expression is at most twice the number of non-end nodes in the final branching program, as needed. 
\end{proof}
\begin{corollary} \label{branchingcorollary}
For any branching complexity measure $\mu_b$ and any function $f$, $\mu_b(f) \leq 130n$
\end{corollary}
\begin{proof}
By Lemma \ref{branchingmeasure} we have that for all $m \geq 1$, $\frac{m\mu_b(f) - m\mu_b(1)}{2} \leq m \cdot b_m(f)$. Using Theorem \ref{construction} and noting that $\mu_b(1) \leq 2$ we obtain that $\mu_b(f) \leq 130n$.
\end{proof}
Finally, we note that every submodular complexity measure $\mu_s$ is a branching complexity measure, so Corollary \ref{branchingcorollary} is a slight generalization of Razborov's result \cite{submodular} (though with a worse constant).
\begin{definition}
A submodular complexity measure $\mu_s$ is a measure on functions which satisfies the following properties
\begin{enumerate}
\item $\forall i, \mu(x_i) = \mu(\neg{x_i}) = 1$
\item $\forall f, \mu(f) \geq 0$
\item $\forall f,g, \mu_s(f \vee g) + \mu_s(f \wedge g) \leq \mu_s(f) + \mu_s(g)$
\end{enumerate}
\end{definition}
\begin{lemma}
Every submodular complexity measure $\mu_s$ is a branching complexity measure.
\end{lemma}
\begin{proof}
Note that 
$$\mu_s(f \vee x_i) + \mu_s(f \wedge x_i) \leq \mu_s(f) + \mu_s(x_i)$$ 
and 
$$\mu_s((f \vee x_i) \wedge \neg{x_i}) + \mu_s((f \vee x_i) \vee \neg{x_i}) = \mu_s(f \wedge \neg{x_i}) + \mu_s(1) \leq \mu_s(f \vee x_i) + \mu_s(\neg{x_i})$$
Combining these two inequalities we obtain that 
$$\mu_s(f \wedge \neg{x_i}) + \mu_s(1) + \mu_s(f \wedge x_i) \leq \mu_s(f) + \mu_s(x_i) + \mu_s(\neg{x_i})$$
which implies that 
$\mu_s(f \wedge \neg{x_i}) + \mu_s(f \wedge x_i) \leq \mu_s(f) + 2 - \mu_s(1) \leq \mu_s(f) + 2$, as needed.
\end{proof}
\section{Conclusion}\label{conclusion}
In this paper, we showed that for any function $f$, there is a branching program computing a doubly exponential number of copies of $f$ which has linear size per copy of $f$. This result shows that in the branching program model, any operation/function can be amortized with sufficiently many copies. This result also disproves a conjecture about nonuniform catalytic computation, rules out certain approaches for proving general lower space bounds, and gives a constructive analogue of Razborov's result \cite{submodular} on submodular complexity measures.

However, the number of copies required in our construction is extremely large. A remaining open problem is to determine whether having a doubly exponential number of copies is necessary or there a construction with a smaller number of copies. Less ambitiously, if we believe but cannot prove that a doubly exponential number of copies is necessary, can we show that a construction with fewer copies would have surprising implications?

\end{document}